\documentclass[prodmode,acmec]{acmsmall}

\usepackage[ruled]{algorithm2e}
\renewcommand{\algorithmcfname}{ALGORITHM}
\SetAlFnt{\small}
\SetAlCapFnt{\small}
\SetAlCapNameFnt{\small}
\SetAlCapHSkip{0pt}
\IncMargin{-\parindent}

\usepackage{latexsym,amsmath,amssymb}
\usepackage{comment}

\acmVolume{X}
\acmNumber{X}
\acmArticle{X}
\acmYear{2012}
\acmMonth{2}

\newcommand*{\forms}[1]{Formula(s)~#1}
\newcommand*{\formsref}[1]{\forms{\eqref{#1}}} %for formula references

\newcommand*{\sectn}[1]{Section~#1}
\newcommand*{\sectnref}[1]{\sectn{\ref{#1}}} %for section references

\newcommand{\set}[1]{\left\{ #1 \right\}}

\newcommand{\inv}[1]{\frac{1}{#1}}
\newcommand{\abs}[1]{\left| #1 \right|}

\newcommand{\reals}{\mathbb{R}}

\newcommand{\calB}{\mathcal{B}}

\newcommand{\cals}{\mathcal{s}}

\newcommand*{\defas}{\ensuremath{\stackrel{\rm def}{=}}}
\newcommand{\opt}{\textsc{opt}}

\newtheorem{observation}[theorem]{Observation}

\newdef{defn}[theorem]{Definition}

\newcommand{\defined}[1]{\emph{#1}}

\newcommand{\NE}{NE} %Nash Equilibrium
\newcommand{\SPE}{SPE} %Subgame Perfect Equilibrium

\DeclareMathOperator{\maxt}{M_2} %The second maximum

{\begin{description}}
{\end{description}}

\begin{document}

% Page heads
\markboth{Gleb Polevoy et al.}{Signalling Competition and Social Welfare}

% Title portion
\title{Signalling Competition and Social Welfare\\
(Working Paper)}
\author{GLEB POLEVOY%
\footnote{Corresponding author: \texttt{glebp@tx.technion.ac.il, geruskin@yahoo.com}}
and RANN SMORODINSKY%
\footnote{\texttt{rann@ie.technion.ac.il}}
\affil{Industrial Engineering and Management, Technion}
MOSHE TENNENHOLTZ%
\footnote{\texttt{moshet@microsoft.com}}
\affil{Industrial Engineering and Management, Technion and Microsoft Research}}

\begin{abstract}
We consider an environment where sellers compete over buyers. All
sellers are a-priori identical and strategically signal buyers about
the product they sell. In a setting motivated by on-line advertising
in display ad exchanges, where firms use second price auctions, a
firm's strategy is a decision about its signaling scheme for a
stream of goods (e.g. user impressions), and a buyer's strategy is a
selection among the firms. In this setting, a single seller will
typically provide partial information and consequently a product may
be allocated inefficiently. Intuitively, competition among sellers
may induce sellers to provide more information in order to attract
buyers and thus increase efficiency. Surprisingly, we show that such
a competition among firms may yield significant loss in consumers'
social welfare with respect to the monopolistic setting. Although we
also show that in some cases the competitive setting yields gain in
social welfare, we provide a tight bound on that gain, which is
shown to be small in respect to the above possible loss.

Our model is tightly connected with the literature on bundling in auctions.
\end{abstract}



\acmformat{Polevoy, G., Smorodinsky, R., Tennenholtz, M.,
Signalling Competition and Social Welfare.}

\begin{bottomstuff}
This work is supported by the joint Microsoft-Technion e-Commerce lab.

%Author's addresses: ??
\end{bottomstuff}

\maketitle

\section{Introduction}

One of the main themes in modern economic theory is that the introduction of competition
increases social welfare. Many of the anti-trust laws in the western world are designed to
motivate and protect such competition. The study of competition is a cornerstone of economic theory
and ranges from the classical models of Cournot
\cite{Cournot} and Bertrand \cite{Bertrand} models, to advanced competing auctions
\cite{McAfee,Ellisoncompete}.

Another topic that is well studied is signaling in situations of asymmetric information. Starting
from Spence \cite{Spence} and Akerlof \cite{Akerlof} seminal work,
this issue has become central to the analysis of strategic interaction.
Starting from the work of Milgrom and Weber
\cite{MilgromWeber} it has become central to the analysis of auctions
and economic mechanism design. In this paper we study signaling schemes in the context of competition.
We introduce a simple game which we refer to as a `signaling competition'.
In this game sellers compete over buyers by signalling on the nature of the goods they sell.
We study how the strategic use of signaling schemes effects the social welfare. In particular we compare the
single seller setting with that of competing sellers. The model we use is
motivated by the multi-billion ad exchanges market.

We adopt a model of Emek et al \cite{EFGT}, motivated by Levin and Milgrom \cite{MilgromLevin}.
We adapt the original model to a multi-seller environment.
In this model each firm is an on-line  publisher, which auctions {\it
ad opportunities} to buyers (advertisers). Any time an internet user accesses a web page
that belongs to the publisher an ad opportunity is created. Ad opportunities are different from one another and
each can be described by a long list of attributes.
Some of these attributes are related to the context of the page (e.g., a page that discusses triathlon events
and hence may appeal to advertisers of sportswear) while others are related to the user
(e.g., a male from Seattle that is a member of a forum on fishing). It is at the publisher's discretion to choose
which attributes to expose to the potential buyers.%
\footnote{There are clearly privacy related issues in this model, which we choose to ignore.}

We consider a finite set of attribute combinations and we refer to each such combination as a {\it variant}
of the product. We assume that each ad opportunity is randomly drawn from some commonly
known distribution over the variants (we restrict attention to a uniform distribution in this work).
Each buyer has a known valuation for each of the possible variants.
More specifically, we restrict attention to a binary model where the value of each variant to a buyer may be
either 0 or 1.
Social welfare is naturally defined as the probability of assigning a random good to a buyer
that values it. To simplify the analysis and isolate the impact of signalling we assume that all sellers
have the same pricing mechanism. In particular, each seller conducts a 2nd-price auction.
Thus, a seller's strategy is what to signal about the good.
A signalling scheme is modeled as a partition of the set of variants.
Note that such a partition is closely related to the notion of bundling which we discuss later.

The {\it Signalling Competition} is therefore a two-stage game, where the sellers choose a signalling scheme
(partition)
in the first stage and buyers choose which seller to access given the publicly announced signalling schemes.
Once buyers and sellers match up a random variant is realized and is allocated according to a second price auction.
In fact such a game captures one instance out of a flow of ad opportunities.
We assume, implicitly, that
the signaling scheme is determined for the whole stream of ad opportunities and that each
buyer's choice of a seller (publisher) is also fixed throughout.
This setting closely resembles situations arising
in display ad exchanges (see Muthukrishnan \cite{Muthukrishnan09} for a discussion).

Our main goal is to compare the social welfare in an environment with a single seller with that of
$k$ sellers. We assume that the number of sellers does not impact the number of ad opportunities.
Therefore, we resort to the following natural notion of welfare -
the probability that a random ad opportunity will be allocated to a buyer that values it.
In case there are $k$ sellers we consider the average probability of efficient allocation among the sellers.
This notion is motivated by the fact  that a seller will only receive $\frac{1}{k}$ of the ad opportunities.

It is intuitive to assume that competition among sellers could increase efficiency as such competition
induces sellers to provide more information  about the product, which in turn increases efficiency.
However, a countervailing force also exists - more sellers imply that each seller faces a
thinner market, which in turn could reduce efficiency. Market thinning is a major concern
in the online advertisement business and is discussed in depth by Levin and Milgrom \cite{MilgromLevin}.

Our findings are surprising: the consumers' social welfare in the competing
setting may severely drop with respect to the monopolistic
setting.  This is not necessarily the case and there are situations where
the social welfare may also be higher in the competing setting.
However, we provide a tight bound on this efficiency gain, which is shown to be
negligible in comparison to the potential loss due to competition.

It worth to notice that the model of signaling about the instance of a probabilistic good is
technically equivalent to a model of bundling of goods.
Indeed, the partition of the possible variants of a good into clusters,
while selling a stream of such instances along time,
can be associated with a fixed bundling of a given set of goods, and selling all bundles,
assuming additive valuations.
The use of such bundling of goods in the context of computational advertising
has been discussed in the CS literature in Ghosh et al \cite{Ghosh},
following the rich literature on the bundling of goods, originated in Palfrey \cite{Palfrey}.
Indeed, the idea that neither bundling
all items together, nor selling them separately is optimal has been made explicit in Jehiel et al \cite{Jehiel}.
Our results can be viewed  as applicable to the setting of bundling and to the above setting of signaling,
interchangeably. Whereas this literature focuses on the single auctioneer setting
we are interested in the implications of competition among sellers.

\begin{comment}
\subsection{Related Work}

TBD

\subsection{Our Contribution}

TBD

First, we bound the social welfare in the case of a
single profit maximizing seller. Then, we consider a competitive setting,
and obtain lower and upper bounds on the ratio between the social welfare
in a subgame perfect equilibrium and an arbitrarily behaving sellers and
users. And the punchline of the paper is comparing the ration between the
social welfare of $S$ competing sellers and a single profit-maximizing
seller. The results suggest that competition is not necessarily a bless,
for it can sometimes achieve a low social welfare, while its ability to
boost the total social welfare is at most $\min\set{S, B} + 1$ times
the case of a single sellers, where $B$ is the total number of buyers in
the market. Since the mere fact that $S$ sellers sell $S$ goods, while a
single seller sells one trivially implies a factor of at lest
$\min\set{S, B}$, the upper bound of $\min\set{S, B} + 1$ may come as a
surprise.
\end{comment}

\subsection{Organization}

The paper is organized as follows: \sectnref{Sec:signal_comp} provides the  model and most of the definitions.
\sectnref{Sec:NE_basic_qualit} proves the existence of semi-pure subgame perfect equilibria, namely subgame perfect equilibria where all the buyers play pure strategies.
\sectnref{Sec:monopoly} focuses on the single seller case (the monopoly) and provides bounds on the social welfare loss for a revenue maximizing seller while
\sectnref{Sec:competition} analyzes the loss for the general case.
\sectnref{Sec:monop_compet} compares the prevailing efficiency in a monopolistic setting with that of a competitive setting (more than one seller).
Finally, \sectnref{Sec:conclusion_further} concludes.

\section{Model} \label{Sec:signal_comp}

\subsection{Definitions}    \label{Sec:signal_comp_defn}
A {\em signalling competition} $SC=SC(S,B,G,V)$, is a game composed of
{\bf (1)}
Two mutually exclusive finite sets of players: the set of
sellers~${S}$ and of the set of
buyers~${B}$.
{\bf(2)}
A finite set, ${G}$, of possible variants of a good. Each seller has a single good for sale.
{\bf (3)}
A valuation matrix,$V=(v^g_b)_{b \in B}^{g\in G}$, where $v_b^g  \in \{0,1\}$ denotes the valuation of buyer $b \in B$ for the variant $g \in G$. Players' utilities
are assumed quasi-linear.

Slightly abusing notation we will
use $S$, $B$ and $G$ to denote the cardinality of the corresponding sets as well.

The signalling competition proceeds as follows:
\begin{enumerate}
    \item First, each seller commits to a partition on the set ${G}$ of
    variants.%
    \footnote{A \defined{partition} on a set $X$ is a set of non-empty
    subsets of $X$ such that each two have an empty intersection and the
    union of all the subsets is $X$.}
    \item Second, given the sellers' partitions, buyers simultaneously choose a seller.
    \item Third, a variant of the good for each seller is realized according
    to a commonly known prior, which we assume is uniform, and variants are independent across sellers. Each seller knows which variant is realized but buyers only know to which element of the
    seller's partition the variant belongs (realistically, sellers disclose this information).
    \item Finally, each seller runs a second price auction and goods are
    allocated to the winners.
\end{enumerate}

A (pure) strategy of a seller $s$ is a partition on the set of
variants, $G$. We will typically denote such a partition of $s$ by
$\pi_s$ and the set of all partitions is denoted $\Pi$. An element $\pi \in \Pi^S$ is a strategy tuple of the sellers. A (pure) strategy for buyer $b$ is a function $f_b\colon \Pi^S \to {S}$. Thus, $f_b(\pi)$ denotes the seller chosen by $b$ for a given strategy profile of the sellers, $\pi \in \Pi^S$.
Buyers bid truthfully and the allocation of the good and its price are determined according a second price auction (note that we assume truthful bidding and ignore, from the outset, weakly dominated strategies).

For a buyer $b$, and a subset of variants $\overline{G} \subseteq G$, let
$v_b(\overline{G}) \defas \sum_{g\in \overline{G}}v_b^g$. Note that $\frac{v_b(\overline{G})}{|\overline{G}|}$  is $b$'s expected valuation conditional on a variant being in $\overline{G}$. Let $v_b \defas v_b(G)$.
Analogically, for a good $g$, and $\overline{B} \subseteq B$, we define the total demand for~$g$ in~$\overline{B}$ by $v^g(\overline{B}) \defas \sum_{b\in \overline{B}}{v_b^g}$ and $v^g \defas v^g(B)$.

For any $\overline{B} \subseteq B, \overline{G} \subseteq G$, denote by
$\bar V_1(\overline{B}, \overline{G}) \defas
\max\set{v_b(\overline{G}) | b\in \overline{B}}$ and let $\hat b$ denote an arbitrary buyer that satisfies $v_{\hat b}({\overline G}) = \bar V_1(\overline{B}, \overline{G})$. Let  $\bar V_2(\overline{B}, \overline{G}) \defas
\max(\set{v_b(\overline{G}) | b\in \overline{B}\setminus{\hat b}})$ be the second highest valuation for goods in $\overline G$ by players in $\overline B$.
We drop the reference to the set of players when we refer to the whole set, $B$.
%and omitting $\overline{B}$ and/or $\overline{G}$ will mean it is the whole $B$ and/or $G$,respectively.

%Preliminary definitions
Given a strategy tuple $\pi=(\pi_s)_{s\in S}$ for the sellers and a strategy tuple $f=(f_b)_{b \in B}$ for the buyers we denote by ${\calB}_s(\pi,f) \defas \set{{\hat b} | f_{\hat b}(\pi) = s}$  the set of buyers that choose $s$ and by
${\calB}_b(\pi,f) \defas \set{{\hat b} | f_{\hat b}(\pi) = f_{b}(\pi)}$  the set of buyers that co-bid with $b$ at the same seller.

%The formulae for the players' profits
Using this notation we can now write an explicit formula for players' utilities.

The utility of buyer $b$ is ${u}_b(\pi,f) = \sum_{{\overline G}\in \pi_{f_b(\pi)}}
\frac{\max\{ v_b({\overline G}) - \bar V_2({\calB}_b(\pi,f),{\overline G}), 0\} }{\abs{{\overline G}}}\cdot \frac{\abs{{\overline G}}}{G}$, and so:

\begin{equation}
{u}_b(\pi,f) =
\inv{G}
\sum_{{\overline G}\in \pi_{f_b(\pi)}}
\max\{ v_b({\overline G}) - \bar V_2({\calB}_b(\pi,f),{\overline G}), 0\}
\label{eq:buyer_utility}
\end{equation}

The utility of seller $s$ is
\begin{equation}\label{eq:utility of seller}
{u}_s(\pi,f) = \sum_{{\overline G}\in \pi_s}
\frac{\bar V_2({\calB}_s(\pi,f),{\overline G}) }{\abs{{\overline G}}}\cdot \frac{\abs{{\overline G}}}{G}
=
\inv{G}
\sum_{{\overline G}\in \pi_s}
\bar V_2({\calB}_s(\pi,f),{\overline G})
%\nonumber
\end{equation}

We extend the notion of utility from pure strategies to mixed strategies in the usual way, via expectation.

\subsection{Social Welfare}

Our model assumes a constant stream of goods that is sold via $S$ competing auctions. Eventually each good is either allocated to a player who values it at $1$ or at $0$. The efficiency measure we consider is the proportion of goods allocated to agents that do value them. Thus, given a strategy tuple $(\pi,f)$ we define the {\em social welfare} as follows:

\begin{equation}
SW(\pi,f)=  \inv{SG}\sum_{s\in S}
\sum_{{\overline G}\in \pi_s}
\bar V_1({\calB}_s(\pi,f),{\overline G}).
\label{eq:social_welfare}
\end{equation}

\begin{lemma}
\label{lem:SW_is_sum_utilties}
$SW(\pi,f) = \inv{S}(\sum_s u_s(\pi,f) + \sum_b u_b(\pi,f)).$
\end{lemma}

\begin{proof}
Note that for each seller $s$:
\begin{eqnarray*}\label{obs:total_prof_hbid}
 u_s(\pi_s,f)+\sum_{b\in {\calB}_s(\pi,f)}u_b(\pi_s,f) = \\
= \inv{G} \sum_{\bar G \in \pi_s} \left ( \bar V_2({\calB}_s(\pi,f),{\overline G}) +
\sum_{b\in {\calB}_s(\pi,f)} \max\{ v_b({\overline G}) - \bar V_2({\calB}_b(\pi,f),{\overline G}), 0\} \right ) = \\
=\inv{G}\sum_{\bar G \in \pi_s}\max\{ v_b({\overline G})| b\in {\calB}_s(\pi,f\}
=
\inv{G}\sum_{\bar G \in \pi_s}\bar V_1({\calB}_s(\pi,f),{\overline G}).
\end{eqnarray*}

Summing over $s \in S$ and dividing by $S$ yields the result.
\end{proof}

The next proposition has already been observed by ~\cite{EvenDarKearnsWortman2007} and its proof is omitted:

\begin{proposition}\label{one_seller_two_buyers_sw}
Assume $S=1$, then as the seller's partition, $\pi_s$, gets finer the social welfare (weakly) increases.%
\footnote{$\pi_s$ is a refinement of ${\hat \pi}_s$ if any set in ${\hat \pi}_s$ is a union of sets in $\pi_s$.}
\end{proposition}

\section{Semi-Pure Subgame Perfect Equilibria} \label{Sec:NE_basic_qualit}

A signalling competition $SC=(B,S,G,V)$ together with a tuple of pure strategies, $\pi$, for the sellers induce a game for the buyers, which we denote by $SC_{\pi}$. A Subgame Perfect Equilibrium~\cite{Selten65} of $SC$ is a (possibly mixed) strategy tuple $(\pi,f)$ such that $f(\tau)$ forms a Nash equilibrium for the game $SC_{\tau}$, for any pure strategy tuple, $\tau$, of the sellers. Furthermore, $\pi_s$ is a best reply to the strategy of the other players $(\pi_{-s},f)$.

A Subgame Perfect Equilibrium is called {\em semi-pure} if buyers use pure strategies.

\subsection{The subgame $SC_\tau$}

We argue that for any pure strategy tuple of the sellers, $\tau$, the induced game $SC_\tau$, is a potential game. Consequently we can conclude it must have a pure Nash equilibrium. In the game $SC_\tau$ each player (buyer) must choose a seller and so the players share the strategy set $S$.  Following ~\cite{MondererShapley1996} we say that $SC_\tau$ is a {\em potential game} if there exists some function
$P \colon S^B \to \reals$ such that for any $s\in S^B$ $\forall b\in B \forall \bar{s_b}, \hat{s_b} \in S$:
\begin{equation}
u_b(\bar s_b, s_{-b}) - u_b(\hat s_b, s_{-b})  = P(\bar s_b, s_{-b}) - P(\hat s_b, s_{-b})
\label{eq:w_poten}
\end{equation}

The function $P$ is called the {\em potential} of the game.

\begin{theorem} \label{The:subgm_buyers_potent}
For any strategy tuple  of the sellers, $\tau$, the game $SC_\tau$ is a potential game and the social welfare is the potential.
\end{theorem}

To prove Theorem ~\ref{The:subgm_buyers_potent} we consider the game when player $b'$ is absent. Formally, this is the game $SC(S,B',G,V')$, where $B'= B\setminus\{b'\}$ and
$V' =(v_b^g)_{b \in B'}^{g\in G}$.

\begin{lemma} \label{prop:the_rest_total_same}
If $S = 1$ then for any $b' \in B$, $\ SW(\pi,f)- SW(\pi,f') = u_{b'}(\pi,f)$,
where  $f' =(f)_{b \not= b'}$.
\end{lemma}

In words, the decrease in social welfare resulting from the absence of buyer $b'$ is equal the utility of $b'$ in her presence.

\begin{proof}
Resorting to equation ~\ref{eq:social_welfare} it is sufficient to prove that:
$$
\sum_{\bar G \in \pi}\bar V_1(B,\bar G) -
\sum_{\bar G \in \pi}\bar V_1(B',\bar G)
=
\sum_{{\bar G}\in \pi}\max\{ v_{b'}({\overline G}) - \bar V_2(B,{\overline G}), 0\}.
$$

Therefore, it is sufficient to show that for any $\bar G \in \pi$:
$$
\bar V_1(B,\bar G) -
\bar V_1(B',\bar G)
=
\max\{ v_{b'}({\overline G}) - \bar V_2(B,{\overline G}), 0\}.
$$

For any $\bar G\in \pi$ one of the following cases holds:
\begin{itemize}
\item
Case 1: $v_{b'}(\bar G) < \bar V_1(B',\bar G)$ in which case
$\bar V_1(B,\bar G) =  \bar V_1(B',\bar G)$ and
$v_{b'}({\overline G}) \le \bar V_2(B',{\overline G})$ and the assertion follows.
\item
Case 2: $v_{b'}(\bar G) \ge \bar V_1(B',\bar G)$
in which case
$
\bar V_2(B,{\overline G}) = \bar V_1(B',\bar G)
$
and
$
\bar V_1(B,{\overline G}) = v_{b'}({\overline G})
$, and the assertion follows.
\end{itemize}
\end{proof}

We return to the proof of Theorem ~\ref{The:subgm_buyers_potent}:

\begin{proof}
Let $\tau$ be a pure strategy of the sellers and let $f$ and $g$ be two strategy tuples for the buyers satisfying: (1) $f_b=g_b \ \forall b \not = b'$;  and (2) $f_{b'}(\tau)\not = g_{b'}(\tau)$.
We shall prove that

$$
u_{b'}(\tau,f)-u_{b'}(\tau,g) = SW(\tau,f)-SW(\tau,g).
$$

Let $B' = B \setminus\{b'\}$ and consider a game with the set of buyers $B'$, with the strategy tuple
$f'=g' = (f_b)_{b \not = b'}$.
By applying Lemma \ref{prop:the_rest_total_same} twice we have:
$$
SW(\pi,f)-u_{b'}(\pi,f) = SW(\pi,f')= SW(\pi,g') = SW(\pi,g)-u_{b'}(\pi,g)
$$
and the assertion follows.
\end{proof}
It is well known that any potential game has a pure \NE, and so the following is immediate:

\begin{corollary}\label{corol:buy_steady}
For any pure strategy tuple of the sellers, $\tau$, the game $SC_{\tau}$ has a pure \NE.
\end{corollary}

This reflects on our model of a signalling competition, as follows:

\begin{theorem}
Any signalling competition contains a semi-pure \SPE{}.
\end{theorem}

\begin{proof}
For any pure strategy of the sellers, $\tau$ set $f(\tau)$ to be the pure Nash equilibrium of $SC_\tau$, which exists according to Corollary~\ref{corol:buy_steady}. Fixing the buyers strategies induces a simultaneous move game for the sellers. Let $\pi$ be the (possibly mixed) Nash equilibrium of this game. We argue that ($\pi,f)$ is the desired semi-pure a \SPE{} of the original game.
\end{proof}

Let $\SPE(SC)$ denote the set of all semi-pure subgame perfect equilibria of $SC$. $\SPE(SC)\not = \emptyset$.

\section{Social Welfare in a Monopoly}  \label{Sec:monopoly}

We now study the single seller case. In this game buyers' strategy set is degenerated and so, in fact, we have single player game.

\begin{proposition}\label{one_seller_two_buyers_util}
Assume $S=1$ and $B=2$ then as the seller's partition, $\pi_s$, gets finer, buyer's utility (weakly) increases and the seller's utility (weakly) decreases.
\end{proposition}

\begin{proof}
Denote $B = \set{a, b}$. Now  ${u}_{a}
= \inv{G}\sum_{\bar G\in \pi_s}{\max\set{v_a(\bar G) - v_b(\bar G), 0}} $
which (weakly) increases as $\pi_s$ get finer.

The seller's utility is
${u}_s = \inv{G}\sum_{\bar G\in \pi_s}{\min\set{v_a(\bar G), v_b(\bar G)}}$,
which (weakly) decreases as $\pi_s$ get finer.
\end{proof}

The claims in Proposition \ref{one_seller_two_buyers_util}
do not hold for the case of more than two buyers. We demonstrate this with the following example:
\begin{example}\label{example_monopoly}
Consider the case of one seller, $3$ buyers and $3$ variants for the good. Let the valuation matrix be:
\begin{equation*}
V =
\begin{bmatrix}
1 & 0 & 1\\
1 & 1 & 0\\
0 & 1 & 1
\end{bmatrix}
\end{equation*}
where the $b^{th}$ row represents buyer $b$ and the $g^{th}$ column represents good $g$.
Consider the partition $\pi = \set{\set{1, 2}, \set{3}}$. If the good
is in the set $\set{1, 2}$, the second buyer gets the good and pays $\frac{1}{2}$. In case the good is in $\set{3}$, either buyer 1 or 3 get the good but her full valuation of 1.
Therefore $u_2(\pi)= \frac{2}{3}\frac{1}{2}+\frac{1}{3}0 = \frac{1}{3}$, the seller earns $u_s(\pi)= \frac{2}{3}\frac{1}{2}+\frac{1}{3}1 =  \frac{2}{3}$ and the social welfare is $SW(\pi)= 1$.

Compare this with the finest partition $\tau = \set{\set{1}, \set{2}, \set{3}}$,
where $u_2(\tau)=0 < u_2(\pi)$, $u_s(\tau)=1 > u_s(\pi)$.
\end{example}

As argued, more information leads to increased efficiency. In some cases (e.g., example
~\ref{example_monopoly}) a revenue-maximizing seller would strategically choose to reveal all the information and maximal efficiency would prevail. Unfortunately, this is not always the case.

Let us denote by $\opt{(SC)} = \max_{(\pi,f)}SW(\pi,f)$ the optimal social welfare possible.

\begin{theorem} \label{The:SW_third_opt}
When $S=1$ the efficiency obtained in any semi-pure SPE is at least $\frac{1}{3}\opt$, and this bound is tight.
\end{theorem}

For the proof we make use of the following notation. For any subsets $\overline{G}
\subseteq G$ and $\overline{B} \subseteq B$ let
$p^i{(\overline{B}, \overline{G})} \defas
\abs{\set{g\in \overline{G}| \sum_{b \in \overline{B}}v_b^g \ge i}}$ denote the number of goods in $\overline{G}$ for which there are at least $i$ buyers in $\overline{B}$ that value them. In particular $p^1{(\overline{B}, \overline{G})}$ is the number of goods for which there is some demand and $p^2{(\overline{B}, \overline{G})}$ is the number of goods for which there are at least two interested buyers. We abbreviate $p^i{(B, \overline{G})}$ by $p^i{(\overline{G})}$.

\begin{proof}
Note it is enough to prove the theorem for any pure strategy of the monopolist in our single player game.

Suppose to the contrary, that for some valuation matrix $V$ there is a revenue maximizing partition, $\pi$, such that $SW(\pi) < \frac{1}{3} \opt(SC)$. In this case there must exist some ${\bar G}\in \pi$, such that
$\frac{{\bar V}_1(\bar G)}{\bar G} <
\frac{1}{3}\cdot \frac{p^1({\bar G})}{\bar G}$.
Without loss of generality we can assume that $\pi$ is degenerated and $\bar G = G$. Therefore
\begin{eqnarray}
p^1({G}) > 3\cdot {\bar V}_1(G).    \label{eq:less_3_opt}
\end{eqnarray}

%We show now, that ${u}_s = \inv{G}\sum_{{\bar G}\in \pi}{\maxt\set{v_b({\bar G}) | b\in B}}$ is not maximized, contradictory to our assumption.

Let ~$a \neq b$ be two distinct buyers satisfying
$v_a({G}) = {\bar V}_1(G)$ and
$v_b({G}) = {\bar V}_2(G)$.
Define ${G}' \defas\set{ g \in G | v_b^g+v_a^g \ge1} \subseteq {G}$,
the set of goods wanted by $a$ or $b$,
and let ${G}'' \defas {G} \setminus {G}'$.
Note that by construction $p^1({G}) = p^1({G}')+ p^1({G}'')$ and
$p^1({G'}) \leq 2{\bar V}_1(G)$. Combining this with inequality
~\ref{eq:less_3_opt} implies $p^1({G}'') > {\bar V}_1({G})$. Hence $p^1({G}'')> {\bar V}_1({G}'')$, which in turn implies that ${\bar V}_2({G}'') > 0$.

Based on equation \ref{eq:utility of seller} we compute the revenue from partitioning $G$ into ${G}'$ and ${G}''$:
\begin{eqnarray*}
\frac{{\bar V}_2({G'})+{\bar V}_2({G}'')}{G}
=   \frac{{\bar V}_2({G})+{\bar V}_2({G}'')}{G}
> \frac{{\bar V}_2({G})}{G},
\end{eqnarray*}

where the equality ${\bar V}_2({G}') = {\bar V}_2({G})$ follows from the construction of ${G}'$. As the right hand side corresponds to the original revenue we have arrived at a contradiction.

Finally, to prove this bound is tight consider the following valuations:
\begin{equation*}
V =
\begin{bmatrix}
1 & 0 & 0\\
0 & 1 & 0\\
0 & 0 & 1
\end{bmatrix}
\end{equation*}

The seller maximizes revenue by disclosing nothing. The ensuing
social welfare is $\frac{1}{3}$, while the optimal welfare, obtained when all information is disclosed, is $\opt = 1$.
\end{proof}

Whereas, Theorem ~\ref{The:SW_third_opt} bounds the welfare in the worst possible outcome, it turns out, as has already been shown by Ghosh et al \cite{Ghosh} that there exists some revenue maximizing partition that yields a better outcome:

\begin{theorem}\label{The:SW_third_opt_and_sw}
Let $S=1$, then there exists some seller strategy, $\pi$, such that  $\pi$ maximizes the seller's  profit
and $SW(\pi) \ge \frac{1}{2}\opt$.
Moreover, this bound is tight.
\end{theorem}

\section{Social Welfare in a Competition}   \label{Sec:competition}

When there is more than a single seller we can also expect less-than-optimal social welfare. In fact, the following example demonstrates that the ratio between welfare in a ~\SPE ~ and optimal welfare can be as low as $\frac{1}{G}$:

\begin{example} \label{examp:low_in_SPE}
There are $S = 2$ sellers, $G$ goods and
$B = 2 \cdot G$ buyers with the following matrix:%
\footnote{This example can be generalized to an arbitrary number of sellers, $S\ge 2$,
 using a valuation matrix which is defined by $S$ vertical repetitions of the $G \times G$ identity matrix.}
\begin{equation*}
V =
\begin{bmatrix}
1           & 0             & \ldots    &   0               &   0\\
0           & 1             & 0             & \ldots    &   0\\
\vdots  &   \cdots  &   \ddots  & \cdots    & \vdots\\
0           & \ldots    & 0             &   1               & 0\\
0           & \ldots    &   0               & 0             &   1\\

1           & 0             & \ldots    &   0               &   0\\
0           & 1             & 0             & \ldots    &   0\\
\vdots  &   \cdots  &   \ddots  & \cdots    & \vdots\\
0           & \ldots    & 0             &   1               & 0\\
0           & \ldots    &   0               & 0             &   1
\end{bmatrix}
\end{equation*}

The optimal social welfare is obtained when the two sellers
disclose everything, buyers $1,\ldots,G$ choose seller 1 and buyers
$G+1,\ldots,2G$ choose seller 2. In this case the social welfare is 1.

On the other hand, consider the following strategy profile:

\begin{itemize}
    \item
    Sellers disclose nothing.
    \item Buyers play as follows:
    \begin{itemize}
        \item
If no seller deviates then buyers $1,\ldots,G$ choose seller 1 and buyers
$G+1,\ldots,2G$ choose seller 2.
    \item
If seller 1 deviates to partition $\pi_1$ then for any $G' \in \pi_1$ let
$B'=\{b\le G:V_b(G')>0\}$ be the set of buyers which assign $g'$ a positive valuation. An arbitrary buyer in $B'$ chooses seller 1 while the other choose seller 2. In addition, buyers $G+1,\ldots,2G$ choose seller 2.
\item
If seller 2 deviates then use a similar profile.
    \item
If the two sellers deviate play an arbitrary pure Nash equilibrium in the induced
        subgame (the existence of which follows immediately from Theorem \ref{The:subgm_buyers_potent}).
        \end{itemize}
\end{itemize}

We argue that this strategy profile constitutes a \SPE: Clearly a unilateral deviation of a seller cannot be profitable.
As to the buyers, the only non-obvious case is when a single seller, say seller 1 without loss of generality, deviates. Note that buyers that have chosen seller 1 will get positive utility whereas deviation  to seller 2 yields a utility of zero. On the other hand any buyer that has chosen player 2 will necessarily get a utility of zero when going to seller 1.

In this \SPE ~ the social welfare is $\frac{1}{G}$ and so the ratio between this and the optimum is $\frac{1}{G}$.
\end{example}

It turns out that $\frac{1}{G}$ is the lower bound on welfare loss if there are sufficiently many buyers:

\begin{theorem} \label{rem:examp_tight}
For $S \geq 2$ sellers, $G$ goods and $B \geq S$ buyers satisfying $v_b(G)>0 \ \forall b\in B$ (in  words, each buyer has a positive demand), $\ \frac{SW(\pi,f)}{OPT} \ge \frac{1}{G}$ for all semi-pure SPE $(\pi,f)$.
\end{theorem}

\begin{proof}
Let $(\pi,f)$ be an arbitrary semi-pure SPE. By the definition of a SPE it follows that for any pure strategy of the sellers, $\tau$, the strategy tuple,
$f(\tau)$, is a Nash equilibrium in the induced game $SC_\tau$.

If  ${\cal B}_{s}(\tau,f) \not = \emptyset$ then
$\sum_{{\overline G}\in \pi_s}\bar V_1({\calB}_s(\pi,f),{\overline G})\ge 1$.

Therefore, if ${\cal B}_{s}(\tau,f) \not = \emptyset$ $\ \forall s \in S$ then
$SW(\pi,f) = \frac{1}{SG}\sum_{s\in S}\sum_{{\overline G}\in \pi_s}\bar V_1({\calB}_s(\pi,f),{\overline G})
\ge \frac{1}{G}$.

If, on the other hand, there exists some seller $\hat s$ such that ${\cal B}_{\hat s}(\tau,f) = \emptyset$ (no buyer chooses seller $\hat s$). The utility of every player $b$ is no less than the utility after deviating to seller $\hat s$, which itself must be at least $\inv{G}$. Therefore $u_b(\tau,f)\ge \inv{G}$. Lemma \ref{lem:SW_is_sum_utilties} implies that $SW(\tau,f) \ge  \inv{S} \cdot B\cdot \inv{G}$. Since
$B \geq S$, this implies that $SW(\tau,f)\ge \frac{1}{G}$. As this holds for any pure strategy $\tau$ it must also hold for any mixture. In particular, $SW(\pi,f)\ge \frac{1}{G}$.

On the other hand, the maximal social welfare possible is no more than $1$, and therefore ratio is at least $\frac{1}{G}$.
\end{proof}

In terms of an upper bound, it is easy to construct examples where in equilibrium the optimal social welfare is obtained (e.g., when $v_b(G)=G$ for all players).

\section{Monopoly versus Competition}   \label{Sec:monop_compet}

We now turn to the central question of this work. Given a valuation matrix, will competition among sellers increase efficiency?  On the one hand competition induces market thinning for each seller. This may result in low demand and inefficient allocation. On the other hand, intuition suggests that competition will induce sellers to reveal more information, thus improving allocation.

Recall that $p^i=p^i(B,G)$ indicates the number of goods which demand is greater or equal $i$. In Theorem ~\ref{The:SW_third_opt} we have shown that for $S=1$, $~SW(\pi,f) \ge \frac{OPT}{3} = \frac{p^1(B,G)}{3G}$. We now improve upon this bound:

\begin{lemma}   \label{lemma:SW_third_opt_and_sw_refined}
Let $S=1$ and assume $\pi=\{G\}$ is an optimal strategy for the seller. Then
$SW(\pi,f) \ge \frac{p^1+ p^2}{3G}$. If, in addition, $\pi=\{G\}$ yields the maximal 
social welfare among all revenue maximizing strategies then
$SW(\pi,f) \ge \frac{p^1+ p^2}{2G}$.
\end{lemma}

\begin{proof}
In the single seller game the buyers have a degenerate strategy and so henceforth we omit the reference to $f$ in the notation.

As $\pi=\{G\}$ the social welfare is $SW(\pi,f) = \frac{\bar V_1(G)}{G}$ and the seller's profit is $\frac{\bar V_2(G)}{G}$.

In case $B=1$ or $p^2=0$ or $G=1$ both parts of the lemma follow trivially 
(note that $B=1 \implies p^2=0$). In case $B=2$ (2 buyers), 
$p^1= {\bar V}_1(G)+{\bar V}_2 - p^2 \le 2{\bar V}_1(G)- p^2$.
Therefore,
$ \frac{p^1 + p^2}{2G}\le \frac{2{\bar V}_1(G)}{2G} =SW(\pi)$, as required

Henceforth, we assume $B\ge 3$, $p^2\ge 1$ and $G \geq 2$.

Define $\hat G \defas \set{g\in G | {v^g(B)\geq 2}}$. There are $p^2$ elements in $\hat G$. For an arbitrary $g \in \hat G$ consider the partition ${\bar \pi} = \{G\setminus\{g\},\{g\}\}$.
Clearly ${\bar V}_2(G)-1 \le {\bar V}_2(G\setminus\{g\}) \le {\bar V}_2(G)$. If ${\bar V}_2(G\setminus\{g\}) = {\bar V}_2(G)$
then the seller's profit at $\bar \pi$ is:
\begin{equation*}
u_s({\bar \pi}) = \frac{1}{G}\cdot 1 + \frac{G-1}{G}\frac{{\bar V}_2(G\setminus\{g\})}{G-1} = \frac{1}{G}\cdot 1 + \frac{G-1}{G}\frac{{\bar V}_2(G)}{G-1} > \frac{{\bar V}_2(G)}{G}
\end{equation*}
contradicting the optimality of $\pi$. Therefore $g\in \hat G  \implies {\bar V}_2(G\setminus\{g\}) = {\bar V}_2(G)-1$. Applying this iteratively for all $p^2$ goods in $\hat G$ implies:
\begin{equation}
{\bar V}_2(G\setminus {\hat G}) = {\bar V}_2(G)-p^2
\label{eq:removal_of_good}
\end{equation}

If, in addition, $\pi$ is the efficient revenue-maximizing strategy then similar arguments
yield 
\begin{equation}
{\bar V}_1(G\setminus {\hat G}) = {\bar V}_1(G)-p^2
\label{eq:removal_for_V_1}
\end{equation}

Let $a\not =b$ be two distinct players such that
$v_a(G) = {\bar V}_1(G)$ and $v_b(G) = {\bar V}_2(G)$.
Let $G' \defas \set{g:v_a(g)+v_b(g)\ge 1} \subseteq G$ be the set of goods demanded by players $a$ or $b$ and let $G'' \defas G \setminus G'$. Note that the revenue from the partition $\hat \pi =\{G',G''\}$ is no less then that of $\pi$, then by the maximality of $\pi$ we must conclude that no revenue is generated from $G''$, or more formally,
${\bar V}_2(G'')=0$. In particular, if player $c \not =a,b$ satisfies $v_c(G'')={\bar V}_1(G'')$, then for any player, $d\not =c$, $v_d(G'')=0$, which in return implies that $\set{g\in G'' | {v^g(B)\geq 2}} = \emptyset$ and so:
\begin{equation}
\label{empty_intersection}
\hat G \subset G'.
\end{equation}

Clearly 
$v_{a}(G\setminus{\hat G})\ge v_{a}(G)-p^2 \ge \bar V_2(G)-p^2 = \bar V_2(G\setminus{\hat G})$ 
and similarly 
$v_{b}(G\setminus{\hat G})\ge \bar V_{2}(G\setminus{\hat G})$.

If 
$\max\set{v_a(G\setminus{\hat G}), v_b(G\setminus{\hat G})} < \bar V_1(G\setminus{\hat G})$ 
then 
$v_{a}(G\setminus{\hat G}) = v_{b}(G\setminus{\hat G}) = \bar V_2(G\setminus{\hat G})=\bar V_2(G)-p^2$. This implies that  $g \in \hat G \iff v_a(g)=v_b(g)=1$
and so
\begin{equation}
p^1(G') = v_a(G')+v_b(G')-p^2.
\end{equation}
This implies
$$
p^1+p^2 = p^1(G')+p^1(G'')+p^2 = [v_a(G')+v_b(G')-p^2] +p^1(G'')+p^2
\le 3 \bar V_1(G).
$$

If, on the other hand, 
$\max\set{v_a(G\setminus{\hat G}), v_b(G\setminus{\hat G})} = \bar V_1(G\setminus{\hat G})$ then $v_c(G\setminus{\hat G})\le \bar V_2(G\setminus{\hat G})$. To see this assume the opposite inequality holds, namely $v_c(G\setminus{\hat G})> \bar V_2(G\setminus{\hat G})$, in which case $\bar V_1(G\setminus{\hat G})> \bar V_2(G\setminus{\hat G})$ and so there are 2 players that value $G\setminus{\hat G}$ strictly more than $\bar V_2(G\setminus{\hat G})$ which is impossible. Hence, $v_c(G\setminus{\hat G})\le \bar V_2(G\setminus{\hat G}) = \bar V_2(G) -p^2$.

As $G'' \subset G\setminus \hat G $ we know that
$p^1(G'') = v_c(G'')\le v_c(G\setminus \hat G) \leq \bar V_2(G)-p^2\le \bar V_1(G)-p^2$.

Therefore,
$$
p^1+p^2 = p^1(G')+p^1(G'')+p^2 \le  [v_a(G')+v_b(G')] + [\bar V_1(G)-p^2]+p^2 =
\le 3 \bar V_1(G).
$$

In both cases $p^1+p^2\le 3 \bar V_1(G)$. Dividing both sides by $3G$ yields
$\frac{p^1+ p^2}{3G} \le SW(\pi,f)$, as argued in the first part of the lemma.

Henceforth assume $\pi$ is the efficient revenue-maximizing strategy. 
In this case no social welfare is generated from $G''$ and so ${\bar V}_1(G'')=0$.
In particular, $v_c(G'')=0$ for any player, $c$.

Applying equation \ref{eq:removal_for_V_1}:
$$v_{a}(G\setminus{\hat G})\ge v_{a}(G)-p^2 = \bar V_1(G)-p^2 = \bar V_1(G\setminus{\hat G})
$$
Implying 
\begin{equation}
\label{eq for V1}
v_{a}(G\setminus{\hat G}) = \bar V_1(G)-p^2 = V_1(G\setminus{\hat G}).
\end{equation}

Similarly $v_{b}(G\setminus{\hat G})\ge \bar V_{2}(G\setminus{\hat G})$.
Since $v_{a}(G\setminus{\hat G}) = V_1(G\setminus{\hat G})$, 
we conclude that 
\begin{equation}
\label{eq for V2}
v_{b}(G\setminus{\hat G})= \bar V_{2}(G\setminus{\hat G})=\bar V_2(G)-p^2.
\end{equation}

From equations \ref{eq for V1},\ref{eq for V2} we conclude $g \in \hat G \iff v_a(g)=v_b(g)=1$.  Therefore
\begin{equation}
p^1=p^1(G') = v_a(G')+v_b(G')-p^2(G') = v_a(G')+v_b(G')-p^2.
\end{equation}
Adding $p^2$ to both sides and dividing by $2G$ yields
$\frac{p^1+ p^2}{2G} \le SW(\pi,f)$, as desired.

\end{proof}

\begin{lemma}   \label{lemma:SW_third_opt_and_sw_refined_part_2}
Let $S=1$ and assume $\pi$ is an optimal strategy for the seller. Then
$SW(\pi,f) \ge \frac{p^1+ p^2}{3G}$. If, in addition, 
$\pi$ yields the optimal social welfare among all such strategies, then
$SW(\pi,f) \ge \frac{p^1+ p^2}{2G}$.
\end{lemma}

Note that $\frac{p^1}{G}$ is actually the maximal social welfare, 
and so the first part of Lemma \ref{lemma:SW_third_opt_and_sw_refined_part_2} implies 
Theorem \ref{The:SW_third_opt} while the second part implies Theorem  \ref{The:SW_third_opt_and_sw}.

\begin{proof}
%The general optimal strategy.
Note that it is enough to prove the theorem for the pure strategies, 
so assume $\pi$ be an optimal pure strategy. For any $\bar G \in \pi$ consider the game 
$SC(S,B, \bar G, (v^g_b)_{b\in B}^{g \in \bar G})$. 
Namely the game where the set of goods is a-priori $\bar G$. In this game the strategy  
$\{\bar G\}$ is optimal for the seller. By Lemma ~\ref{lemma:SW_third_opt_and_sw_refined},
part~$1$:

\begin{equation}
\frac{\bar V_1(\bar G)}{\abs{\bar G}} \geq \frac{p^1(\bar G)+ p^2(\bar G)}{3\abs{\bar G}}.
\end{equation}

Averaging over $\bar G \in \pi$:
\begin{eqnarray*}
SW(\pi) = \sum_{\bar G\in \pi_s}{\frac{\abs{\bar G}}{G} \cdot \frac{\bar V_1(\bar G)}{\abs{\bar G}}}
\ge
 \inv{G}\sum_{\bar G\in \pi_s}\frac{p^1(\bar G)+ p^2(\bar G)}{3}
 =
 \frac{p^1(G)+ p^2(G)}{3G}
\end{eqnarray*}

In case, $\pi$ is the efficient revenue maximizing strategy, Lemma ~\ref{lemma:SW_third_opt_and_sw_refined},
part~$2$, yields:
\begin{equation}
\frac{\bar V_1(\bar G)}{\abs{\bar G}} \geq \frac{p^1(\bar G)+ p^2(\bar G)}{2\abs{\bar G}}.
\end{equation} 

Averaging over $\bar G \in \pi$:
\begin{eqnarray*}
SW(\pi) = \sum_{\bar G\in \pi_s}{\frac{\abs{\bar G}}{G} \cdot \frac{\bar V_1(\bar G)}{\abs{\bar G}}}
\ge
 \inv{G}\sum_{\bar G\in \pi_s}\frac{p^1(\bar G)+ p^2(\bar G)}{2}
 =
 \frac{p^1(G)+ p^2(G)}{2G}
\end{eqnarray*}

\end{proof}

Fix a set of buyers, $B$, a set of goods, $G$ and a valuation matrix $v$. 
Let $\Gamma^1$ be the set of optimal strategies for the for the seller in 
$SC(1,B,G,v)$ and let $\Gamma^S$ be the set of all \SPE ~ in $SC(S,B,G,v)$. 
In addition, let $\Gamma'^1$ be the set of optimal strategies for the for the seller in 
$SC(1,B,G,v)$, which yield the maximum social welfare possible there.

\subsection{Upper bound}

We begin by identifying an upper bound on the social welfare attainable in a competition.
Let $c_1=p^1-p^2$ be the number of goods with a unit demand.

\begin{lemma}
\label{lem:max_SW_in_comp}
Let $(\pi,f)$ be an arbitrary strategy tuple in the game $SC(S,B,G,V)$. Then
$SW(\pi,f) \le \frac{c_1+min\set{S,B}\cdot p^2}{SG}$.
\end{lemma}

\begin{proof}
Clearly the maximal welfare can be obtained when sellers fully reveal their information. That is we can assume, W.L.O.G that $\pi=\{\{g\}: g \in G\}$. Let $\hat G = \{g:\sum_b v^b_g=1\}$.  By equation
\ref{eq:social_welfare}:
$$
SW(\pi,f) = \inv{SG}\sum_{s\in S}
\sum_{g \in G}
\bar V_1({\calB}_s(\pi,f),\{g\}) =
$$
$$
=
\inv{SG}\left (
\sum_{g \in \hat G}\sum_{s\in S}
\bar V_1({\calB}_s(\pi,f),\{g\})+ \sum_{g: \sum_b v_b^g \ge 2}\sum_{s\in S}
\bar V_1({\calB}_s(\pi,f),\{g\})\right ) \le
$$
$$\le
\inv{SG} (\sum_{g \in \hat G} 1 + p^2 min\set{S,B}) = \frac{c_1+min\set{S,B}\cdot p^2}{SG}.
$$
\end{proof}

\begin{theorem} \label{The:spe_prof_max_up_bound}
For any $B,G,v$ satisfying $B\ge 2$, for any $S\ge 2$, any
$\hat \pi\in \Gamma^1$ and any $(f,\pi) \in \Gamma^S$:
$$\frac{SW(\pi,f)}{SW(\hat \pi)} \leq 1+\inv{ S}.$$
\end{theorem}

In words, a competition can result in negligible higher efficiency.

\begin{proof}
We define the auxiliary variable $\rho \defas \frac{p^2}{c_1}$. We separate the proofs to prove to the cases $\rho \leq 1$ and $\rho > 1$.

{\bf Case 1 - $\rho \leq 1$ and  $c_1>0$:} By Lemma~\ref{lemma:SW_third_opt_and_sw_refined_part_2}, $SW(\hat \pi) \ge \frac{c_1 + 2p^2}{3G}$ and by Lemma \ref{lem:max_SW_in_comp}
$SW(\pi,f) \le \frac{c_1+min\set{S,B}\cdot p^2}{SG}$.
Therefore the ratio satisfies:
\begin{eqnarray*}
\frac{SW(\pi,f)}{SW(\hat \pi)}  \le \frac{\frac{\min\set{S,B}\cdot p^2 + c_1}{SG}}
{\frac{c_1 + 2p^2}{3G}}
= \frac{3}{S}\frac{c_1 + \min\set{S,B}p^2}{c_1 + 2p^2}
= \frac{3}{S}\frac{1 + \min\set{S,B}\rho}{1 + 2\rho}
\end{eqnarray*}

If $\min\set{S,B} = 2$, then  $\frac{SW(\pi,f)}{SW(\hat \pi)}  \le \frac{3}{S}\frac{1 + 2\rho}{1 + 2\rho} = \frac{3}{S} \leq
\frac{1 + S}{S} = 1+\inv{S}$, as claimed. If $\min\set{S,B} > 2$, the function $g(\rho) = \frac{1 + \min\set{S,B}\rho}{1 + 2\rho}$ is
strictly increasing, and obtains its maximum at $\rho=1$. Therefore,
$\frac{SW(\pi,f)}{SW(\hat \pi)}  \le \frac{3}{S}\frac{1 + \min\set{S,B}}{1 + 2} \le 1+\inv{S}$, as claimed.

{\bf Case 2 - $\rho > 1$ and $c_1>0$:}
We now use an alternative lower bound on a profit maximizing single seller's social welfare, which is $p^2/G$. %Note that bounds the welfare for any strategy of the seller, not only for revenue maximizing ones. - WRONG, thus commented out!
Therefore:

\begin{eqnarray*}
\frac{SW(\pi,f)}{SW(\hat \pi)}  \le \frac{\frac{\min\set{S,B}\cdot p^2 + c_1}{SG}}
{\frac{p^2}{G}} = \inv{S}({\min\set{S,B} + \frac{c_1}{p^2}})
=  \inv{S}( {\min\set{S,B} + \inv{\rho}}) \le 1+\inv{S}.
\end{eqnarray*}

{\bf Case 3 - $c_1=0$:} Using the same lower bound on social welfare as in Case 2 yields:
\begin{eqnarray*}
\frac{SW(\pi,f)}{SW(\hat \pi)}  \le \frac{\frac{\min\set{S,B}\cdot p^2 + c_1}{SG}}
{\frac{p^2}{G}} = \frac{\min\set{S,B}}{S} < 1+\inv{S}.
\end{eqnarray*}

\end{proof}

Note that the last theorem, although stated for SPE actually holds for any strategy profile in the competition setting ($S\ge 2$). This bound turns out to be a tight bound, even if we restrict strategy profiles in the competition to SPE:

\begin{theorem}
For any $S\ge 2$  there exists a set of buyers, $B$, and valuations such that
there are strategy tuples $\hat \pi\in \Gamma^1$ and $(f,\pi) \in \Gamma^S$,
satisfying:
$\frac{SW(\pi,f)}{SW(\hat \pi)}  = 1+\inv{S}$.
\end{theorem}

\begin{proof}
For any~$S \geq 2$, set $B = S + 1, G = 2$ and consider the following valuation matrix:
\begin{equation*}
V =
\begin{bmatrix}
S \left\{
\begin{matrix}
0 & 1\\
\cdots & \cdots\\
0 & 1
\end{matrix}\right.\\
\begin{matrix}
1 & 0
\end{matrix}
\end{bmatrix}
\end{equation*}
A single profit maximizing seller may use the strategy $\hat \pi =\{1,2\}$ (i.e., disclose nothing), yielding $SW(\hat \pi) = 1/2$.

For the setting with $S$ sellers, the following constitutes a SPE:
\begin{itemize}
\item
Sellers strategy: - All the sellers disclose everything.
\item
Buyers strategy:
\begin{itemize}
\item
If no seller deviates from the above prescribed strategy then buyer $i$ goes to seller $i$, for $i<S$ and buyers $S,S+1$ go to seller $S$.
\item
If one seller deviates, say seller $k<S$ then then buyer $i$ goes to seller $i$, for $i<S$ and buyers $S,S+1$ go to seller $S$.
\item
If seller $S$ deviates buyer $S$ moves to seller $1$.
\item
If at $2$ sellers or more deviate then an arbitrary equilibrium profile of the induced subgame is played.
\end{itemize}
\end{itemize}

The social welfare obtained in this strategy tuple is
$\inv{S}((S-1)\frac{1}{2} + 1\cdot 1) =\frac{S+1}{2S}$, therefore the ratio is $1+\inv{S}$.
\end{proof}

Recall that Theorem~\ref{The:SW_third_opt_and_sw} provides a different lower bound on the social welfare obtained by a single seller, assuming she chooses a social welfare maximizing strategy, among all the revenue maximizing ones. In such a case:

\begin{theorem}
For any $B,G,v$ satisfying $B\ge 2$, for any $S\ge 2$, any
$\hat \pi\in \Gamma'^1$ and any $(f,\pi) \in \Gamma^S$:
$$\frac{SW(\pi,f)}{SW(\hat \pi)} \leq \frac{\min\set{S,B}}{S} \leq 1.$$
\end{theorem}
In words, a competition cannot result in a higher efficiency in this case!

\begin{proof}
By Lemma~\ref{lemma:SW_third_opt_and_sw_refined_part_2}, $SW(\hat \pi) \ge \frac{c_1 + 2p^2}{2G}$ and by Lemma \ref{lem:max_SW_in_comp}
$SW(\pi,f) \le \frac{c_1+min\set{S,B}\cdot p^2}{SG}$.
Therefore the ratio satisfies:
\begin{eqnarray*}
\frac{SW(\pi,f)}{SW(\hat \pi)}  \le \frac{\frac{\min\set{S,B}\cdot p^2 + c_1}{SG}}
{\frac{c_1 + 2p^2}{2G}}
= \frac{2}{S}\frac{c_1 + \min\set{S,B}p^2}{c_1 + 2p^2}
\leq \frac{2}{S}\frac{c_1 \frac{\min\set{S,B}}{2} + \min\set{S,B}p^2}{c_1 + 2p^2}\\
= \frac{\min\set{S,B}}{S} \leq 1
\end{eqnarray*}
\end{proof}

Note that the last theorem, although stated for SPE actually holds for any strategy profile in the competition setting ($S\ge 2$). This bound turns out to be a tight bound, even if we restrict strategy profiles in the competition to SPE:

\begin{theorem}
For any $S\ge 2$  there exists a set of buyers, $B$, and valuations such that
there are strategy tuples $\hat \pi\in \Gamma'^1$ and $(f,\pi) \in \Gamma^S$,
satisfying:
$\frac{SW(\pi,f)}{SW(\hat \pi)}  = \frac{\min\set{S,B}}{S}$.
\end{theorem}

\begin{proof}
For any~$S = B \geq 2$, consider the case where every buyer wants each
good, i.e.~$\forall b\in B, g\in G: v_b^g = 1$.
Here, $\hat \pi\in \Gamma'^1$ can be $\set{g}_{g\in G}$, yielding social
welfare~$1$.

Consider the \SPE, where $\forall\pi: f_i(\pi) = i$, and all the sellers
disclose everything, i.e.~$\pi = \set{g}_{g\in G}$. This yields the social
welfare~$1$. Therefore, the ratio is $1 = \frac{\min\set{S,B}}{S}$.
\end{proof}

\subsection{Lower Bound}

\begin{theorem}
If $S \geq 2$, $B \geq S$ and $v_b(G)>0\ \forall b \in B$ (in words, each buyer has a positive demand), then for for any
$\hat \pi\in \Gamma^1$ and $(f,\pi) \in \Gamma^S$:
$$\frac{SW(\pi,f)}{SW(\hat \pi)} \ge \frac{1}{G},$$
and this bound is tight.
\end{theorem}

\begin{proof}
As in the proof of the
Theorem~\ref{rem:examp_tight}, we obtain that $SW(\pi,f) \geq \inv{G}$.

On the other hand, $SW(\hat \pi) \leq 1$. Thus,
$\frac{SW(\pi,f)}{SW(\hat \pi)} \ge \frac{1}{G}$.

To demonstrate that the bound is tight we re-visit Example~\ref{examp:low_in_SPE}.
 In that example, the social welfare in the ~\SPE{} is $\frac{1}{G}$, while a profit maximizing single
seller may disclose everything, yielding a social welfare of $1$. Thus, the ratio is $\frac{1}{G}$.
\end{proof}

\section{Conclusion} \label{Sec:conclusion_further}
%Conclusion
 We analyzed the impact on social welfare due to the introduction of competition
 in a binary valued model of signaling.
 As more sellers compete for buyers we expect sellers to reveal more information
 in order to attract buyers, thus resulting in increased welfare. On the other hand,
 it also causes  market thinning  for any given seller, and so could decrease social welfare.
 As we show the potential increase is negligible while the potential decrease is catastrophical.

% Appendix
%\appendix
%\section*{APPENDIX}
%\setcounter{section}{1}

%\appendixhead{??}

% Bibliography
\bibliography{bibnew}

\bibliographystyle{acmsmall}
%\bibliography{../sellers_compete_buyers}

% History dates
\received{Month YYYY}{Month YYYY}{Month YYYY}

% Electronic Appendix
%\elecappendix

%\medskip

%\section{This is an example of Appendix section head}

\end{document}